\documentclass[preprint,12pt]{elsarticle}

\usepackage{amsmath,amssymb,amsthm}
\usepackage{booktabs,graphicx,url}
\usepackage[colorlinks=true,linkcolor=blue,citecolor=blue,urlcolor=blue]{hyperref}

\newcommand{\Z}{\mathbb{Z}}
\newcommand{\Cay}{\mathrm{Cay}}

\newtheorem{proposition}{Proposition}
\newtheorem{remark}{Remark}

\journal{Signal Processing Open}

\begin{document}

\begin{frontmatter}

\title{Signal-Aware Cayley Completion: Choosing Between an Exact and a Lossy
Abelian Host}

\author[lesia]{Rigobert Fokam Souop\corref{cor1}}
\ead{fokamrigobert@gmail.com}
\cortext[cor1]{Corresponding author.}
\author[lesia,garoua]{Laurent Bitjoka}
\ead{laurent.bitjoka@univ-garoua.cm}
\address[lesia]{Laboratory of Energy, Signal, Imaging and Automation (LESIA),
University of Ngaound\'er\'e, P.O.\ Box 454, Ngaound\'er\'e, Cameroon}
\address[garoua]{Laboratory of Scientific Artificial Intelligence and Applied
Mathematics, University of Garoua, P.O.\ Box 346, Garoua, Cameroon}

\begin{abstract}
Group-embedding graph signal processing buys exact Fourier analysis either by
enlarging the domain until a graph embeds isometrically into an abelian Cayley
graph, or by perturbing the graph until it is one. We give a linear-time
procedure that chooses between the two, and study the criterion governing the
lossy branch. We prove that the perturbation of the Dirichlet form is exactly
the signed sum of the edited-edge energies, and is bounded in modulus by their
total, into which the combinatorial cost of the completion does not enter. We
confirm experimentally that this energy predicts signal fidelity strongly
though not strictly monotonically (Spearman $\rho=-0.79$ on an instance where
the completion, and hence the cost, is held fixed), while the cost itself
carries no information about it. Selecting among completions of identical cost
by this energy recovers $61$ percent of the available performance gap on the
exhaustive range, on a sample too small to establish that the effect persists
beyond eight vertices. On real payment data the router takes the exact branch
everywhere; on real anti-money-laundering data the completion invariant is a
re-encoding of maximum degree and is not distinguishable from it, and the
linear-time bound it rivals is shown to have a blind spot on near-complete
graphs.
\end{abstract}

\begin{keyword}
Graph signal processing \sep  Cayley graph \sep  abelian group \sep  Dirichlet energy \sep  graph editing \sep  anti-money laundering.
\end{keyword}

\end{frontmatter}

\section{Introduction}

Exact harmonic analysis on a graph can be bought in two ways, and they charge
in different currencies.

The isometric route \cite{fokam-p1,fokam-p2,fokam-p3,fokam-p4} leaves the
graph untouched and enlarges the domain: one embeds $G$ isometrically into a
Cayley graph $\Cay(\Gamma,S)$ of a finite abelian group, on which the Fourier
transform, convolution theorem and translation operator are exactly those of
classical harmonic analysis. The metric is preserved perfectly and the price
is host order, which for an irregular graph can reach $2^{n-1}$.

The completion route \cite{fokam-comp} fixes the host order at $n$ and
perturbs the graph instead, editing as few edges as possible until it
\emph{is} an abelian Cayley graph on its own vertices. Nothing is enlarged;
the price is paid in perturbation.

Classical signal processing has taken the second route for decades without
naming it. A finite non-periodic signal of $n$ samples is a signal on the path
$P_n$, and processing it by circular convolution is exactly the statement that
$P_n$ has been completed to $C_n=\Cay(\Z_n,\{\pm1\})$ by one added edge. The
cost is the familiar wraparound artefact, and it is not paid uniformly: it
depends on how different the two endpoint samples are.

This paper turns that observation into an operational procedure and tests it.
We ask which of the two routes to take for a given graph, and, when the
completion route is taken, which completion to choose. The two questions have
different characters, and it is worth separating them at the outset: the
routing decision depends only on the structure of the graph and is made in
linear time, whereas the choice among completions is signal-aware, and it is
there that the title's adjective applies. Our findings are:

\begin{itemize}
\item the combinatorial cost $\gamma$ of a completion \emph{does not} predict
      its effect on a signal, and the Dirichlet energy of the edited edges
      does (Section~\ref{sec:border});
\item among completions of \emph{identical} cost, choosing the one of least
      edited-edge Dirichlet energy recovers $61.3\%$ of the available
      performance gap over choosing arbitrarily, at no search cost
      (Section~\ref{sec:select});
\item on real payment data the linear-time router sends every component to the
      isometric branch, where the host is provably minimal
      (Section~\ref{sec:paysim}); while on real anti-money-laundering data the
      completion invariant is a re-encoding of maximum degree and buys nothing
      over a linear-time bound (Section~\ref{sec:aml});
\item that same linear-time bound has a blind spot which only real data
      exposed: it cannot distinguish a graph that is cheap because it is
      well structured from one that is cheap because it is complete
      (Section~\ref{sec:blindspot}).
\end{itemize}

The last two are negative results. We report them in full because they
delimit where the framework applies, which is more useful than a selective
account of where it succeeds.

\begin{figure}[htbp]
\centering
\includegraphics[width=\linewidth]{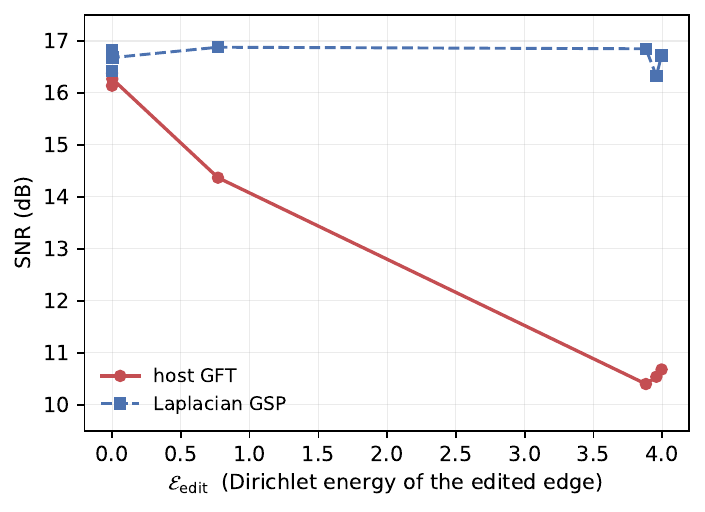}
\caption{Border locality on $P_{64}\to C_{64}$. The completion, and hence
$\gamma^{+}=1/63$, is identical at every point; only the signal varies. Host
performance falls with the Dirichlet energy of the single edited edge while
the Laplacian baseline stays flat. The association is strong but not strictly
monotone (Spearman $\rho=-0.79$); the two rightmost points recover slightly.
Values are those of Table~\ref{tab:border}.}
\label{fig:border}
\end{figure}

\begin{figure}[htbp]
\centering
\includegraphics[width=\linewidth]{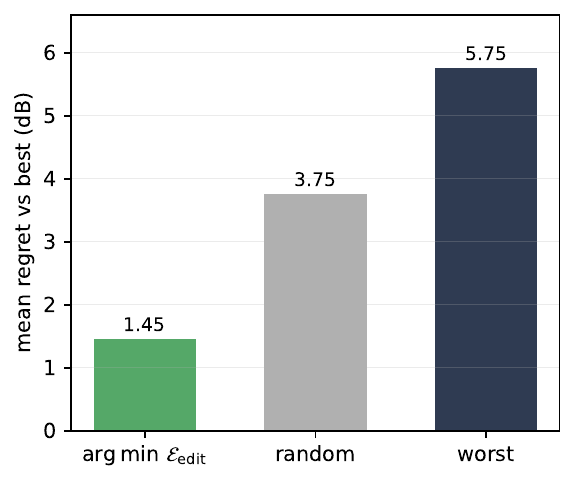}
\caption{Selecting among minimum-cost completions, where $\gamma$ and the
number of edited edges are constant by construction, so only the location of
the edits varies. Choosing the least edited-edge energy leaves a mean regret
of $1.45$~dB against the best available completion, against $3.75$~dB for an
arbitrary choice. Values are those of Table~\ref{tab:select}.}
\label{fig:select}
\end{figure}

\section{Related work}
\label{sec:related}

\subsection{Graph signal processing} The spectral framework built on the graph
Laplacian is surveyed in \cite{shuman2013,ortega2018}, with the adjacency- or
shift-based formulation of \cite{sandryhaila-moura2013} as the alternative
algebraic route; spectral wavelet constructions appear in
\cite{hammond2011}, and sampling theory for bandlimited graph signals in
\cite{chen2015,anis2016}. Spectral background is standard \cite{chung1997}.
The distinguishing feature of the group-embedding approach
\cite{fokam-p3,fokam-p4,fokam-diss} is that the transform on the host is
literally a Fourier transform of a finite abelian group, so Parseval, the
convolution theorem and a unitary translation hold exactly rather than
approximately.

\subsection{Circulant structure} The closest prior line builds shift-invariant
filtering, sampling and critically sampled filterbanks on circulant graphs and
proposes decomposing arbitrary graphs into combinations of circulant pieces
\cite{ekambaram2013a,ekambaram2013b,ekambaram2015}. The present work differs
in object and in aim: a single minimal-edit abelian host of arbitrary group
structure, rather than a decomposition into cyclic components, together with
an invariant quantifying the edit and a criterion for choosing among edits.

\subsection{Isometric embedding into structured hosts} That a graph embeds
isometrically into a hypercube is a classical question
\cite{firsov1965,djokovic1973,winkler1984}, surveyed in
\cite{ovchinnikov2008}, with the extension to Hamming graphs and products in
\cite{wilkeit1990,graham-winkler1985,imrich-klavzar2000,hammack2011} and the
addressing problem of \cite{graham-pollak1971} resolved in
\cite{winkler1983}. Scale embeddings and $\ell_1$-graphs are treated in
\cite{shpectorov1993,deza-shpectorov1996,laurent1994,deza-laurent1997}, with
recognition complexity in \cite{imrich-klavzar1996,aurenhammer1995},
lattice dimension in \cite{eppstein2005}, retracts in \cite{bandelt1984},
distance-hereditary graphs in \cite{bandelt-mulder1986}, recent Hamming
embeddings of weighted graphs in
\cite{berleant2023,sheridan2021,ebrahimi2025}, and a broad survey in
\cite{bandelt-chepoi2008}. This literature enlarges the host to achieve
exactness; the completion branch used here instead perturbs the graph
\cite{fokam-comp}, and the two are the endpoints of one trade-off
\cite{fokam-p1,fokam-p2}.

\subsection{Algebraic background} Cayley graphs and vertex-transitive graphs
are standard \cite{godsil-royle2001}, with vertex-transitivity strictly weaker
than being a Cayley graph \cite{mckay-praeger1994}. The extremal input to the
star bounds is the theory of maximum sum-free sets in abelian groups
\cite{diananda-yap1969,rhemtulla-street1970,green-ruzsa2005}, with
enumerative refinements in \cite{alon-balogh-morris-samotij2014}; the Smith
normal form underlying the quotient labeling is surveyed in
\cite{stanley2016}. The census population of \cite{fokam-comp} is enumerated
in the ordering of \cite{read-wilson1998}.

\section{Preliminaries}
\label{sec:prelim}

We recall the two invariants this paper uses; both are developed in
\cite{fokam-comp,fokam-p2} and are restated here so that the present paper is
self-contained.

Let $G$ be finite, connected and simple with $n$ vertices and $m$ edges, and
let $\Gamma$ be a finite abelian group with connection set $S=-S$,
$0\notin S$, generating $\Gamma$.

\emph{Isometric host order.} $\nu(G)$ is the least $|\Gamma|$ for which $G$
embeds isometrically into $\Cay(\Gamma,S)$ for some $S$. Exact values used
below: $\nu(K_{1,q})=2q$, with minimal host the complete bipartite graph
$K_{q,q}$, and $\nu(P_k)=2(k-1)$ \cite{fokam-p2}. Since the star value carries
the only positive real-data result of this paper
(Section~\ref{sec:paysim}), we recall why it holds, in a form that makes that
section self-contained. If $K_{1,q}$ embeds in $\Cay(\Gamma,S)$ with the
centre at $0$, the $q$ leaves lie in $S$ and are pairwise non-adjacent, so no
difference of two of them lies in $S$; that is, the leaf set is a sum-free
subset of $\Gamma$ of size $q$. A sum-free subset of a finite abelian group
has size at most $|\Gamma|/2$ \cite{diananda-yap1969}, whence $|\Gamma|\ge2q$.
Conversely the odd residues of $\Z_{2q}$ form a sum-free set of size $q$, and
taking $S$ to be that set embeds $K_{1,q}$ isometrically, leaves being at
distance exactly $2$ through the centre. Hence $\nu(K_{1,q})=2q$ exactly. The
correspondence between induced stars and sum-free sets is due to Babai and
S\'os \cite[Prop.~7.4]{babai-sos1985}.

\emph{Cayley edit distance.} For $|\Gamma|=n$ and a bijection
$\pi:V(G)\to\Gamma$, write $E_\pi(G)$ for the image edge set. Then
\[
  \gamma_{\triangle}(G)=\frac1m\min_{\Gamma,\pi,S}
  \bigl|E_\pi(G)\,\triangle\,E(\Cay(\Gamma,S))\bigr|,
\]
the analogous minimum over completions that only add edges being
$\gamma^{+}(G)$. Thus $\gamma_{\triangle}(G)=0$ exactly when $G$ is an abelian
Cayley graph. Both are computed exactly for $n\le7$ by exhaustive search over
groups, labelings and connection sets, and by local search beyond
\cite{fokam-comp}.

\emph{Degree bound.} Every abelian Cayley graph is regular, so
$\gamma^{+}(G)\ge \mathrm{tri}(G):=\max\bigl(0,n\Delta/2m-1\bigr)$, computable
in linear time from the degree sequence \cite{fokam-comp}.

\paragraph{Signal processing on Cayley graphs} The closest published work is
that of Beck, Ghandehari, Hudson and Paltenstein \cite{beck2024frames}, who
give a representation-theoretic spectral decomposition for weighted Cayley
graphs and construct Frobenius--Schur and Cayley frames on them; related
Gabor-type constructions appear in \cite{ghandehari2021gabor}. That line takes
the Cayley structure as given and builds the transform; the present paper
takes a graph that has no such structure and quantifies what it costs to give
it one. The two are complementary, and the question asked here --- which of
two ways of acquiring the structure to use, and at what price to the signal
--- does not arise in their setting.

\paragraph{Editing a graph before processing a signal on it} Proposition~1 is
a statement about how the Dirichlet form moves under edge modification, and
that is the object bounded by the graph-reduction and sparsification
literature. Loukas \cite{loukas2019reduction} gives spectral and cut
guarantees for coarsening, and Spielman and Srivastava
\cite{spielman2011sparsification} sample edges by effective resistance so as
to preserve $x^{\top}Lx$ uniformly. Those guarantees are uniform over a
subspace; Proposition~1 is an identity for each signal separately, which is
weaker in scope and exact rather than approximate. The practical consequence
we draw --- select a graph modification by the energy the signal carries
across it, not by the number of edges it touches --- applies to that whole
family and not only to Cayley completion.

\section{The two branches and the router}
\label{sec:router}

Let $G$ be finite, connected, simple, with $n$ vertices, $m$ edges, maximum
degree $\Delta$ and cycle rank $\mu=m-n+1$.

\paragraph{Branch 1: isometric embedding} Take it when the certified host
order is tractable. For trees this is now settled exactly:
$\nu(K_{1,q})=2q$ and $\nu(P_k)=2(k-1)$ \cite{fokam-p2}, the star host being
the complete bipartite graph $K_{q,q}$. There is no perturbation and no
distortion; Parseval, convolution and translation are exact.

\paragraph{Branch 2: completion} Take it when the isometric host is
intractable. Fix the order at $n$ and pay twice: once in edge perturbation,
measured by the Cayley edit distance $\gamma_{\triangle}$
\cite{fokam-comp}, and once in signal fidelity, measured below.

\paragraph{Branch 3: neither} When the completion is also too expensive, no
group-embedding method applies and one should fall back on matrix graph signal
processing \cite{sandryhaila-moura2013}.

The threshold is best set on fidelity rather than on $\gamma$, since
Section~\ref{sec:border} shows $\gamma$ does not predict fidelity. Given a
signal class with second moments $\mathbb{E}|x_u-x_v|^2$ and a tolerance
$\eta$ on the relative perturbation of the Dirichlet form,
Proposition~\ref{prop:exact} makes the rule explicit: take Branch 2 only when
\[
  \frac{\mathbb{E}\,\mathcal{E}_{\mathrm{edit}}(x)}
       {\mathbb{E}\,x^{\top}L_Gx}\;\le\;\eta ,
\]
and Branch 3 otherwise. Both quantities are computed directly from the edited
set and the signal statistics, with no spectral computation.

We instantiate the rule, since a threshold left unspecified is not a
procedure. Take the smooth signal class of Section~\ref{sec:select}, for which
$\mathbb{E}|x_u-x_v|^2$ is approximately equal across edges; the ratio then
reduces to $\mathbb{E}\,|F|/m$ weighted by the relative energy on the edited
edges.

Setting $\eta=0.25$ --- a quarter of the signal's own Dirichlet energy
displaced by the edit --- reproduces the routing of
Table~\ref{tab:router} and sends Zachary to Branch~3, consistently with
Section~\ref{sec:limitations}. We offer $\eta=0.25$ as an order of magnitude
calibrated on these instances and not as a universal constant.

Note that the signal-free quantity $2\max_v\deg_F(v)$ bounds eigenvalue
movement in absolute terms and is not a substitute for the ratio above, which
is dimensionless; the two answer different questions and we do not exchange
them.

\paragraph{The router} Cycle rank and the degree bound
$\mathrm{tri}(G)=\max(0,n\Delta/2m-1)$ of \cite{fokam-comp} are both
computable in linear time and decide the branch.

\begin{table}[t]
\centering
\caption{The routing decision on representative inputs, by the structural
criterion of Section~\ref{sec:router} alone. Zachary is routed to
completion here and re-examined in Section~\ref{sec:limitations}, where accounting
for fidelity moves it to Branch~3; the discrepancy is the point of the paper
rather than an inconsistency, since the structural router does not see the
signal. Host orders for the
isometric branch are the proved minima of \cite{fokam-p2}.}
\label{tab:router}
\begin{tabular}{lrrrrr}
\toprule
graph & $n$ & $\mu$ & $\Delta$ & $\mathrm{tri}$ & branch\\
\midrule
$K_{1,20}$ & $21$ & $0$ & $20$ & $9.500$ & isometric, $\nu=40$\\
$P_{20}$ & $20$ & $0$ & $2$ & $0.053$ & isometric, $\nu=38$\\
$C_{20}$ & $20$ & $1$ & $2$ & $0.000$ & already Cayley\\
Zachary & $34$ & $45$ & $17$ & $2.705$ & completion, order $34$\\
\bottomrule
\end{tabular}
\end{table}

Zachary's karate club is the case that motivates the whole construction: the
isometric pipeline reports $k=33$, a host of order $2^{33}>8\times10^9$,
whereas completion fixes the order at $34$. Section~\ref{sec:limits} returns
to what that reduction costs.

\section{Why $\gamma$ cannot be the criterion}
\label{sec:border}

Fix a completion and let $F$ be its edited edge set. For a signal $x$ define
the \emph{edited-edge Dirichlet energy}
\[
  \mathcal{E}_{\mathrm{edit}}(x)=\sum_{\{u,v\}\in F}\bigl(x_u-x_v\bigr)^2 .
\]

We compare two denoising pipelines on the same noisy signal: an ideal
low-pass in the Laplacian eigenbasis of $G$, and an ideal low-pass in the host
eigenbasis, both retaining the same number of coefficients, so any difference
comes from the basis and not from a looser filter.

\begin{table}[t]
\centering
\caption{$P_{64}$ completed to $C_{64}$ by the single edge $\{0,63\}$. The
completion, and hence $\gamma^{+}=1/63$, is identical in every row; only the
signal changes. Signals $x_f(t)=\cos(2\pi ft)$, $12$ coefficients retained,
$\sigma=0.25$, $40$ trials.}
\label{tab:border}
\begin{tabular}{rrrr}
\toprule
$f$ & $\mathcal{E}_{\mathrm{edit}}$ & SNR host & SNR Laplacian\\
\midrule
$1.0$ & $0.000$ & $16.38$ & $16.82$\\
$2.0$ & $0.000$ & $16.14$ & $16.42$\\
$3.0$ & $0.002$ & $16.27$ & $16.68$\\
$1.25$ & $0.770$ & $14.37$ & $16.88$\\
$2.5$ & $3.881$ & $10.40$ & $16.85$\\
$1.5$ & $3.957$ & $10.54$ & $16.33$\\
$0.5$ & $3.995$ & $10.68$ & $16.72$\\
\bottomrule
\end{tabular}
\end{table}

Table~\ref{tab:border} and Figure~\ref{fig:border} give the central fact. The
host performance falls from $16.38$ to $10.40$~dB as
$\mathcal{E}_{\mathrm{edit}}$ grows, a swing of about $6$~dB and a strong but
not strictly monotone association (Spearman $\rho=-0.79$; the two largest
energies recover by $0.3$~dB), while $\gamma$ is constant throughout. \emph{The combinatorial cost of a completion carries
no information about its effect on a signal.} What predicts the effect is
where the edited edges fall relative to the signal.

A second reading of the same table should be stated plainly, since it
constrains what the framework may claim. At
$\mathcal{E}_{\mathrm{edit}}=0$ the host does not beat the Laplacian; it ties
it, to within half a decibel. On the path there is no denoising advantage to
be had. The value of an abelian host is structural --- exact Parseval, an
exact convolution theorem, a unitary translation operator, and an
$O(n\log n)$ transform \cite{fokam-p3} --- and not an improvement in
mean-squared error. Completion should be understood as the price of admission
to that structure, and $\mathcal{E}_{\mathrm{edit}}$ as the price tag.

\paragraph{A best case, with its baseline} On the $8\times8$ grid completed
to the torus, a periodic-friendly image gives $23.13$~dB on the host against
$6.96$~dB for the Laplacian, a gain of $16.17$~dB. The margin is real but it
is partly a weak-baseline effect: the Laplacian does badly on that image, and
on the same grid the host \emph{loses} by $12.8$~dB on a linear gradient.
Differences of SNR are therefore not monotone in
$\mathcal{E}_{\mathrm{edit}}$, because the baseline itself varies by $20$~dB.
The monotone quantity is the host's own performance, as in
Table~\ref{tab:border}.

\subsection{Why this quantity, and not another}
\label{sec:theory}

The monotonicity above is not a coincidence of the experiment: the edited-edge
Dirichlet energy is exactly the perturbation of the Dirichlet form.

\begin{proposition}\label{prop:exact}
Let $H$ be a completion of $G$ on the same vertex set, with $F_{+}$ the added
and $F_{-}$ the deleted edges, and let $L_G,L_H$ be the Laplacians. Then for
every signal $x$,
\[
  x^{\top}L_H x-x^{\top}L_G x
  =\!\!\sum_{\{u,v\}\in F_{+}}\!\!(x_u-x_v)^2
  \;-\!\!\sum_{\{u,v\}\in F_{-}}\!\!(x_u-x_v)^2 ,
\]
and consequently
$\bigl|x^{\top}(L_H-L_G)x\bigr|\le\mathcal{E}_{\mathrm{edit}}(x)$.
\end{proposition}

\begin{proof}
$x^{\top}Lx=\sum_{\{u,v\}\in E}(x_u-x_v)^2$ for any graph Laplacian. Applying
this to $H$ and to $G$ and subtracting leaves exactly the edges in which they
differ, with sign according to whether the edge was added or deleted. The
inequality follows from the triangle inequality, since
$\mathcal{E}_{\mathrm{edit}}$ sums the same terms with all signs positive.
\end{proof}

So if $x$ is smooth on $G$, in the sense $x^{\top}L_Gx\le\lambda\|x\|^2$, then
$x^{\top}L_Hx\le\lambda\|x\|^2+\mathcal{E}_{\mathrm{edit}}(x)$: the completion
moves the signal's energy in the host spectrum by at most
$\mathcal{E}_{\mathrm{edit}}$, and by exactly the signed sum above. This is
why $\gamma$ cannot serve. It counts $|F_{+}|+|F_{-}|$, which does not appear
in the identity at all, whereas the quantity that does appear is precisely the
one we measure.

A signal-free companion bounds the spectrum itself: writing $\deg_F(v)$ for
the number of edited edges at $v$, we have
$\|L_H-L_G\|_2\le 2\max_v \deg_F(v)$, so by Weyl's inequality every Laplacian
eigenvalue moves by at most that amount. Edits that spread over many vertices
are spectrally gentle; edits concentrated on a hub are violent at any
$\gamma$. Both statements were verified numerically on random instances.

What remains open is the step from these first-order statements to a bound on
the error of the convolution theorem itself, of the form
$\|\widehat{f*h}-\hat f\hat h\|\le C\,\mathcal{E}_{\mathrm{edit}}\|f\|\|h\|$.
Proposition~\ref{prop:exact} controls the quadratic form but not the
eigenvectors, and it is the eigenvector perturbation that enters a filtering
error; a Davis--Kahan argument would require a spectral gap that a general
completion does not provide. We regard this as the main theoretical gap.

\paragraph{Robustness to the filter} The decrease is not an artefact of the
ideal low-pass. Repeating the $P_{64}$ experiment with a Tikhonov filter
$1/(1+\tau\lambda)$ and a heat kernel $e^{-\tau\lambda}$ gives host SNR
$17.14\to15.34\to10.10$ dB and $16.97\to15.37\to11.12$ dB respectively across
$\mathcal{E}_{\mathrm{edit}}=0,\,0.77,\,3.88$, against
$16.38\to14.37\to10.40$ for the ideal filter, the values of
Table~\ref{tab:border}. All three decrease across these three points. Over the
full seven-signal range of Table~\ref{tab:border} the association is strong
but not monotone (Spearman $\rho=-0.79$); we report it as such.

\section{Choosing among equally cheap completions}
\label{sec:select}

If $\gamma$ does not discriminate, something must. We test whether
$\mathcal{E}_{\mathrm{edit}}$ does, on a design in which $\gamma$ is
prevented from discriminating by construction.

Among the $965$ connected graphs on six and seven vertices, $587$ admit at
least three \emph{distinct} minimum-cost completions. Since a completion on
$V(G)$ is determined by its edited set $F$, distinct $F$ means a distinct
host. Every competitor therefore has the same $\gamma$, and, because the edit
cost \emph{is} $|F|$, the same number of edited edges. Only the location of
the edits varies.

Signals are drawn smooth with respect to $G$ itself, as random combinations of
its three lowest non-trivial Laplacian eigenvectors: the realistic situation
in which a signal is smooth on the true graph and the question is which
completion preserves that smoothness. Denoising uses an ideal low-pass in the
host Laplacian eigenbasis.

\begin{table}[t]
\centering
\caption{Selecting among minimum-cost completions, $120$ (graph, signal)
instances over $40$ graphs; $\gamma$ and $|F|$ are constant within each
instance. Regret is measured against the best completion available.}
\label{tab:select}
\begin{tabular}{lr}
\toprule
quantity & value\\
\midrule
median Spearman $\rho(\mathcal{E}_{\mathrm{edit}},\text{SNR})$ & $-0.646$\\
instances with $\rho<0$ & $89.2\%$\\
mean regret, select $\arg\min\mathcal{E}_{\mathrm{edit}}$ & $1.452$ dB\\
mean regret, select at random & $3.754$ dB\\
mean regret, worst case & $5.753$ dB\\
gap over random recovered & $61.3\%$\\
\bottomrule
\end{tabular}
\end{table}

Of the $965$ connected graphs on six and seven vertices, $587$ admit at least
three distinct minimum-cost completions. From these we retained $40$
(a uniform random sample) and drew three signals on each, giving the $120$ instances of
Table~\ref{tab:select}. Proportions from $120$ instances carry a Wilson
$95\%$ interval of roughly $\pm 6$ points: the reported $89.2\%$ of instances
with $\rho<0$ has interval $[82.3\%,93.6\%]$, so the sign of the association
is established while its exact rate is not.

The criterion works (Table~\ref{tab:select}, Figure~\ref{fig:select}). It is
stable: over nine combinations of filter bandwidth and noise level, $\rho$ is
negative and $\arg\min\mathcal{E}_{\mathrm{edit}}$ beats random selection in
all nine. It is insensitive to noise and sensitive to bandwidth, being
strongest when the filter width matches the bandwidth of the signal --- which
is itself informative, since the criterion cannot see frequencies the filter
admits but the signal does not carry.

It is not optimal. A residual mean regret of $1.45$~dB against the best
available completion remains, so a better predictor exists; identifying it is
open.

\paragraph{The criterion weakens with size, on a small sample} Because the
population above is the exhaustive census on six and seven vertices, we ran a
controlled scaling test: the same random generator $G(n,p)$ with $p\in[0.28,0.5]$, the same
protocol, and only $n$ changing, with minimum-cost completions still
enumerated exactly.

\begin{table}[htbp]
\centering
\caption{Controlled scaling test on $8$ instances per column. Same generator,
same protocol, only $n$ changes; minimum-cost completions enumerated exactly
in both rows. Wilson $95\%$ intervals on the proportions are
$[40.9\%,92.9\%]$ and $[30.6\%,86.3\%]$ respectively: with this sample the two
columns are not statistically distinguishable, and the row is reported as an
indication, not as a measurement.}
\label{tab:scaling}
\begin{tabular}{lrr}
\toprule
& $n=7$ & $n=8$\\
\midrule
median $\rho$ & $-0.688$ & $-0.252$\\
instances & $8$ & $8$\\
instances with $\rho<0$ & $75.0\%$ \ ($6/8$) & $62.5\%$ \ ($5/8$)\\
gap over random recovered & $60.2\%$ & $13.0\%$\\
\bottomrule
\end{tabular}
\end{table}

The point estimates degrade substantially, and the design isolates size as the
only varying factor (Table~\ref{tab:scaling}); but with eight instances per
column the proportions have overlapping confidence intervals, and only the
drop in median $\rho$ and in recovered gap is large relative to the sample.
We therefore state the claim only
where it is verified: on the exhaustive range, selecting by
$\mathcal{E}_{\mathrm{edit}}$ is a cheap and reliable improvement over an
arbitrary choice; at $n=8$ it is still an improvement, but a small one, and we
have no evidence that it survives to the graph sizes of practice.

The likely reason is visible in Proposition~\ref{prop:exact}: the identity
controls the quadratic form, not the eigenvectors, and as $n$ grows the
competing hosts differ increasingly in their eigenvector structure rather than
only in where their edits fall. Establishing the criterion at scale, or
replacing it by one that accounts for the eigenvectors, is in our view the
most important open problem this paper leaves.

\paragraph{Generator selection is a threshold, not a gradient} The same
principle applied to the choice of connection set on a circulant host of order
$60$ gives $24.82$~dB for $\mathcal{E}=1.09$, $5.39$ and $14.86$ alike;
$10.69$~dB at $\mathcal{E}=25.75$; and $-0.01$~dB at $\mathcal{E}=241.09$.
Performance is flat while the energy stays low and collapses past a threshold.
The practical consequence is that $\mathcal{E}_{\mathrm{edit}}$ is best used
as an admissibility filter rather than as a ranking.

\section{Experimental setup}
\label{sec:setup}

\emph{Synthetic experiments.} Signals are normalised to unit mean square and
corrupted with white Gaussian noise of the stated standard deviation; every
reported figure is a mean over the stated number of independent noise
realisations. Denoising retains the same number of spectral coefficients in
both bases, so no comparison is won by a looser filter. Low-pass filtering on
a host $\Cay(\Gamma,S)$ keeps the characters of smallest Laplacian eigenvalue
$\lambda_k=\sum_{s\in S}\bigl(1-\cos\langle k,s\rangle\bigr)$; note that
keeping the lowest \emph{indices} of a discrete Fourier transform is correct
only when $S=\{\pm1\}$, and gives badly misleading results otherwise.

\emph{PaySim.} From the transaction log we keep customer-to-customer transfers
and build the undirected simple graph on account identifiers. Merchant
counterparties, whose identifiers carry a distinguishing prefix, are dropped:
they are degree-one leaves and would dominate every degree statistic. Balance
columns are excluded throughout as a documented source of label leakage.
Statistics are reported on the first $20$ simulation steps.

\emph{IBM AML HI-Small.} The labelled file consists of blocks delimited by
begin and end markers naming the typology, each block listing transaction
lines. We parse each block into an undirected simple graph, taking node
identity to be the (bank, account) \emph{pair} rather than the account string
alone, since the same account identifier recurs under different banks;
self-loops are dropped and the free-text description following the typology
name is discarded.

\emph{The co-counterparty projection.} Two accounts are joined when they share
at least one neighbour in the account graph. Vertices of degree above $200$
are skipped, since a single hub would otherwise contribute a clique on its
whole neighbourhood.

\section{Real data I: payment networks route to the exact branch}
\label{sec:paysim}

PaySim \cite{lopezrojas2016} is a mobile-money simulator whose transaction log
we reduce to an undirected account-to-account graph, merchants dropped as
degree-one leaves. On steps $1$--$20$: $410{,}677$ accounts, $357{,}729$
edges, mean degree $1.74$, $52{,}948$ components, and among the $19{,}392$
components with at least eight vertices, \emph{none contains a cycle}.

The largest component has $94$ vertices and maximum degree $93$. Being
acyclic and connected on $94$ vertices, it is a tree, hence has exactly $93$
edges. A vertex of degree $93$ in a graph on $94$ vertices is adjacent to
every other vertex, and a tree containing such a vertex has no further edges,
so the component is exactly $K_{1,93}$; no other graph is consistent with the
three reported numbers. By \cite{fokam-p2} its minimal isometric
host has order $\nu=2q=186$, with excursion ratio
$\varepsilon=n/N=94/186=0.505$.

Account-level payment graphs are therefore forests of stars. Cycle rank is
zero throughout, the router sends every component to the isometric branch, and
each is solved exactly with a provably minimal host. No completion, no
perturbation and no distortion is incurred anywhere in the dataset. This is
the framework's clean positive case, and it is a consequence of the structure
of the data rather than of any tuning.

\section{Real data II: a negative result on laundering typologies}
\label{sec:aml}

The IBM AML HI-Small dataset \cite{altman2023} contains $370$ labelled
laundering sub-graphs across eight typologies. We parse each block, taking
node identity to be the (bank, account) pair, and compute
$\gamma_{\triangle}$ exactly for the $232$ instances within exhaustive range
and by search beyond it.

The invariant orders the typologies sensibly: median $\gamma_{\triangle}$ is
$0.000$ for CYCLE and BIPARTITE, $0.125$ for RANDOM, $0.500$ for STACK,
$1.071$ for SCATTER-GATHER, $1.571$ for FAN-OUT and $1.625$ for FAN-IN. On the
binary task of separating hub-dominated typologies from the rest it attains
$85.4\%$ accuracy with no false positives.

\begin{table}[t]
\centering
\caption{IBM AML HI-Small, $370$ labelled instances: separating hub-dominated
typologies from the rest. Accuracy is the best attainable by a single
threshold on each predictor.}
\label{tab:aml}
\begin{tabular}{lr}
\toprule
predictor & best accuracy\\
\midrule
majority class & $50.5\%$\\
cycle rank $\mu$ & $58.6\%$\\
mean degree $2m/n$ & $69.5\%$\\
order $n$ & $71.9\%$\\
density & $73.8\%$\\
\midrule
$\gamma_{\triangle}$ (edit distance, requires search) & $85.4\%$\\
degree bound $\mathrm{tri}(G)$ (linear time) & $85.4\%$\\
maximum degree $\Delta$ (linear time) & $\mathbf{86.5\%}$\\
\bottomrule
\end{tabular}
\end{table}

The linear-time degree bound attains \emph{exactly the same} accuracy, and
maximum degree by itself does slightly better still
(Table~\ref{tab:aml}); $\gamma_{\triangle}$ correlates with $\Delta$ at
$r=0.872$ (Table~\ref{tab:aml}). Against a majority-class baseline of
$50.5\%$, all three are informative, but the ordering is unambiguous: on this
dataset the completion invariant is a re-encoding of $\Delta$, and the solver
buys nothing that the degree sequence does not already give --- indeed
slightly less. The practical takeaway for practitioners is direct: on
hub-dominated transaction graphs, do not run the completion solver, compute
the degree sequence.

The other invariants we tried separate the classes considerably less well,
which indicates that the degree bound is not succeeding merely because the
task is easy.

Two artefacts of the data are worth recording for anyone reusing it. Of the
$54$ blocks labelled CYCLE, $14$ contain no cycle once the transaction
directions are symmetrised, because two-hop cycles collapse. BIPARTITE and
STACK blocks are frequently disconnected, at $63\%$ and $60\%$ of blocks
respectively.

\section{A blind spot in the linear-time bound}
\label{sec:blindspot}

Faced with the sparsity of the account graph, the natural response is to
change the graph model rather than the solver, and to project accounts that
share a counterparty. On PaySim steps $1$--$5$ with merchants included, the
co-counterparty projection has $2{,}413$ nodes and $18{,}838$ edges, and
\emph{all $92$} of its components with at least eight vertices pass the
completion triage.

A pass rate of $100\%$ looks like success and is vacuous.

\begin{proposition}\label{prop:proj}
The co-counterparty projection of a disjoint union of stars is a disjoint
union of complete graphs on the leaf sets, together with one isolated vertex
per centre; the projection of $K_{1,q}$ is $K_q$ plus an isolated vertex.
\end{proposition}

\begin{proof}
Two vertices are joined in the projection when they share a neighbour. In
$K_{1,q}$ every pair of leaves shares the centre, and the centre shares no
neighbour with anything, so the projection is exactly the complete graph on
the $q$ leaves.
\end{proof}

The transcript's largest projected component, with $n=46$ and $\Delta=45$, is
therefore $K_{46}$. And the degree bound cannot see this: for any regular
graph $n\Delta/2m-1=0$, so complete graphs sit at the very bottom of the
triage scale. Indeed $\gamma_{\triangle}(K_q)=0$, correctly, because
$K_q=\Cay(\Z_q,\Z_q\setminus\{0\})$ already is a Cayley graph.

\begin{remark}
The linear-time bound cannot distinguish a graph that is cheap because it is
well structured from one that is cheap because it is complete. It is a lower
bound on cost, and a complete graph costs nothing --- correctly, but
uselessly. Any deployment must screen near-complete inputs before trusting a
low triage score.
\end{remark}

This limitation was not visible on the census of \cite{fokam-comp}, where
complete graphs are a negligible fraction of the population. It appeared only
on real data, and it belongs alongside the bound it qualifies.

\section{Limitations}
\label{sec:limitations}
\label{sec:limits}

\paragraph{Completion relocates cost, it does not remove it} Zachary's karate
club reduces from a host of order $2^{33}$ to one of order $34$, but a
cyclic-host completion at a modest search budget requires $70$ edits on $78$
edges, $\gamma_{\triangle}\approx0.897$. The graph is tractable in host size
and severe in perturbation. Whether the trade is acceptable is exactly the
question $\mathcal{E}_{\mathrm{edit}}$ is meant to answer, and for a graph
this heavily edited the correct routing decision is Branch 3.

\paragraph{No error-propagation bound} We give an empirical criterion, not a
theorem. The bridge that is missing is a bound of the form
$\|\widehat{f*h}-\hat f\hat h\|\le C(\delta)\|f\|\|h\|$ relating metric
distortion of the host to error in the convolution theorem. The companion
paper \cite{fokam-comp} shows that edit count and bi-Lipschitz distortion are
independent invariants, moving in opposite directions on stars and paths;
turning the present Dirichlet criterion into a bound is the natural next step.

\paragraph{Scope of the experiments} The selection criterion is verified on
the exhaustive range and degrades markedly at $n=8$
(Table~\ref{tab:scaling}); we make no claim that it survives to large graphs.
The PaySim conclusions are derived from the reported run statistics rather
than recomputed end to end. The denoising comparisons use ideal, Tikhonov and
heat-kernel filters; we have not tested learned or data-adaptive filters.

\section{Conclusion}

We have given an operational two-branch procedure for group-embedding graph
signal processing, routed in linear time, and tested the criterion that
governs the lossy branch. The combinatorial cost of a completion does not
predict its effect on a signal; the Dirichlet energy of the edited edges does,
monotonically, and selecting by it recovers most of the available gap at no
search cost. On real payment data the router takes the exact branch
everywhere, where the minimal host is now known in closed form. On real
laundering data the invariant adds nothing to a linear-time degree bound, and
that bound in turn has a blind spot on near-complete graphs.

The framework's value, on this evidence, is structural exactness on the graphs
that admit small hosts, together with a linear-time verdict on the graphs that
do not. It is not a general improvement in denoising performance, and the
experiments here do not support claiming one.

Three directions follow. The first is the error-propagation bound discussed in
Section~\ref{sec:theory}: Proposition~\ref{prop:exact} controls the quadratic
form exactly, and what is needed is the corresponding control of eigenvectors,
presumably under a spectral-gap hypothesis that identifies which completions
admit such a bound. The second is scaling: the selection criterion is strong
on the exhaustive range and weak already at $n=8$, and either a demonstration
that it recovers at scale or a replacement that accounts for eigenvector
perturbation would materially extend the method's reach. The third is to use
$\mathcal{E}_{\mathrm{edit}}$ inside the completion search rather than after
it, minimising a weighted objective in $\gamma$ and
$\mathcal{E}_{\mathrm{edit}}$ jointly, which would replace post-hoc selection
among optima by direct optimisation of the quantity that matters. A smaller
practical item is automatic detection of the near-complete inputs of
Section~\ref{sec:blindspot}, for which the density of the complement is an
obvious and cheap flag.

\section*{Declaration of competing interests}
The authors declare that they have no known competing financial interests or
personal relationships that could have appeared to influence the work reported
in this paper.

\section*{Funding}
This research received no specific grant from any funding agency in the public, commercial, or not-for-profit sectors.

\section*{Declaration of generative AI and AI-assisted technologies in the
manuscript preparation process}
During the preparation of this work the authors used Claude Anthropic (Models: Opus and Sonnet) in order to improve the language and
readability of the manuscript and to assist with the organisation of its
content. After using this tool the authors reviewed and edited the content as
needed and take full responsibility for the content of the published article.

\section*{Data availability}
The Author deposited the experiment scripts and derived data on Zenodo doi:10.5281/zenodo.21852006 . Both datasets used are public and are cited in the references  as \cite{lopezrojas2016} and \cite{altman2023}.


\begin{thebibliography}{99}

\bibitem{alon-balogh-morris-samotij2014}
N.~Alon, J.~Balogh, R.~Morris, W.~Samotij, Counting sum-free sets in
abelian groups, Israel J.~Math. 199 (2014) 309--344.

\bibitem{aurenhammer1995}
F.~Aurenhammer, J.~Hagauer, Recognizing binary Hamming graphs in
$O(n^2\log n)$ time, Math. Systems Theory 28 (1995) 387--395.

\bibitem{bandelt-mulder1986}
H.-J. Bandelt, H.~M. Mulder, Distance-hereditary graphs,
J.~Combin. Theory Ser.~B 41 (1986) 182--208.

\bibitem{bandelt1984}
H.-J. Bandelt, Retracts of hypercubes, J.~Graph Theory 8 (1984) 501--510.

\bibitem{bandelt-chepoi2008}
H.-J. Bandelt, V.~Chepoi, Metric graph theory and geometry: a survey, in:
J.~E. Goodman, J.~Pach, R.~Pollack (Eds.), Surveys on Discrete and
Computational Geometry: Twenty Years Later, Contemp. Math. 453,
Amer. Math. Soc., Providence, 2008, pp.~49--86.

\bibitem{berleant2023}
J.~Berleant, K.~Sheridan, A.~Condon, V.~Vassilevska Williams, M.~Bathe,
Isometric Hamming embeddings of weighted graphs,
Discrete Appl. Math. 332 (2023) 119--128.

\bibitem{ebrahimi2025}
J.~Ebrahimi Boroojeni, M.~Oghbaei Bonab, Binary stretch embedding of
weighted graphs, Des. Codes Cryptogr. 93 (2025) 2741--2760.

\bibitem{chung1997}
F.~R.~K. Chung, Spectral Graph Theory, CBMS Regional Conference Series in
Mathematics 92, Amer. Math. Soc., Providence, 1997.

\bibitem{deza-laurent1997}
M.~Deza, M.~Laurent, Geometry of Cuts and Metrics, Springer, Berlin, 1997.

\bibitem{deza-shpectorov1996}
M.~Deza, S.~Shpectorov, Recognition of the $\ell_1$-graphs with complexity
$O(nm)$, or football in a hypercube, European J.~Combin. 17 (1996) 279--289.

\bibitem{babai-sos1985} L.~Babai, V.~T. S\'os, Sidon sets in groups and
induced subgraphs of Cayley graphs, European J.~Combin. 6 (1985) 101--114.
\bibitem{diananda-yap1969}
P.~H. Diananda, H.~P. Yap, Maximal sum-free sets of elements of finite
groups, Proc. Japan Acad. 45 (1969) 1--5.

\bibitem{djokovic1973}
D.~\v{Z}. Djokovi\'c, Distance-preserving subgraphs of hypercubes,
J.~Combin. Theory Ser.~B 14 (1973) 263--267.

\bibitem{eppstein2005}
D.~Eppstein, The lattice dimension of a graph, European J.~Combin. 26 (2005)
585--592.

\bibitem{firsov1965}
V.~V. Firsov, Isometric embedding of a graph in a Boolean cube,
Cybernetics 1 (1965) 112--113.

\bibitem{godsil-royle2001}
C.~Godsil, G.~Royle, Algebraic Graph Theory, Graduate Texts in Mathematics
207, Springer, New York, 2001.

\bibitem{graham-pollak1971}
R.~L. Graham, H.~O. Pollak, On the addressing problem for loop switching,
Bell System Tech.~J. 50 (1971) 2495--2519.

\bibitem{graham-winkler1985}
R.~L. Graham, P.~M. Winkler, On isometric embeddings of graphs,
Trans. Amer. Math. Soc. 288 (1985) 527--536.

\bibitem{green-ruzsa2005}
B.~Green, I.~Z. Ruzsa, Sum-free sets in abelian groups,
Israel J.~Math. 147 (2005) 157--188.

\bibitem{hammack2011}
R.~Hammack, W.~Imrich, S.~Klav\v{z}ar, Handbook of Product Graphs, 2nd ed.,
CRC Press, Boca Raton, 2011.

\bibitem{hammond2011}
D.~K. Hammond, P.~Vandergheynst, R.~Gribonval, Wavelets on graphs via
spectral graph theory, Appl. Comput. Harmon. Anal. 30 (2011) 129--150.

\bibitem{imrich-klavzar1996}
W.~Imrich, S.~Klav\v{z}ar, On the complexity of recognizing Hamming graphs
and related classes of graphs, European J.~Combin. 17 (1996) 209--221.

\bibitem{imrich-klavzar2000}
W.~Imrich, S.~Klav\v{z}ar, Product Graphs: Structure and Recognition,
Wiley, New York, 2000.

\bibitem{laurent1994}
M.~Laurent, Hypercube embedding of distances with few values, in:
H.~Barcelo, G.~Kalai (Eds.), Jerusalem Combinatorics~'93,
Contemp. Math. 178, Amer. Math. Soc., Providence, 1994, pp.~179--207.

\bibitem{mckay-praeger1994}
B.~D. McKay, C.~E. Praeger, Vertex-transitive graphs which are not Cayley
graphs, I, J.~Austral. Math. Soc. Ser.~A 56 (1994) 53--63.

\bibitem{ortega2018}
A.~Ortega, P.~Frossard, J.~Kova\v{c}evi\'c, J.~M.~F. Moura,
P.~Vandergheynst, Graph signal processing: overview, challenges, and
applications, Proc. IEEE 106 (5) (2018) 808--828.

\bibitem{ovchinnikov2008}
S.~Ovchinnikov, Partial cubes: structures, characterizations, and
constructions, Discrete Math. 308 (2008) 5597--5621.

\bibitem{read-wilson1998}
R.~C. Read, R.~J. Wilson, An Atlas of Graphs, Clarendon Press, Oxford, 1998.

\bibitem{rhemtulla-street1970}
A.~H. Rhemtulla, A.~P. Street, Maximal sum-free sets in finite abelian
groups, Bull. Austral. Math. Soc. 2 (1970) 289--297.

\bibitem{sandryhaila-moura2013}
A.~Sandryhaila, J.~M.~F. Moura, Discrete signal processing on graphs,
IEEE Trans. Signal Process. 61 (7) (2013) 1644--1656.

\bibitem{sheridan2021}
K.~Sheridan, J.~Berleant, M.~Bathe, A.~Condon, V.~Vassilevska Williams,
Factorization and pseudofactorization of weighted graphs, arXiv:2112.06990,
2021.

\bibitem{shpectorov1993}
S.~V. Shpectorov, On scale embeddings of graphs into hypercubes,
European J.~Combin. 14 (1993) 117--130.

\bibitem{shuman2013}
D.~I. Shuman, S.~K. Narang, P.~Frossard, A.~Ortega, P.~Vandergheynst,
The emerging field of signal processing on graphs: extending
high-dimensional data analysis to networks and other irregular domains,
IEEE Signal Process. Mag. 30 (3) (2013) 83--98.

\bibitem{fokam-p3}
R.~Fokam Souop, L.~Bitjoka, Group-embedding graph Fourier transforms:
exact spectral analysis on abelian hosts, companion manuscript,
arXiv:2607.13338 [eess.SP], submitted to
IEEE Trans. Signal Inform. Process. Netw., 2026.

\bibitem{fokam-diss}
R.~Fokam Souop,  Minimal Isometric Embeddings of Graphs into
Abelian Groups: Theory, Algorithms, and Applications to Signal Processing over Networks, Ph.D. dissertation, University of Ngaoundéré, 2026, arXiv:2606.29391 [math.CO], 2026.

\bibitem{fokam-p1}
R.~Fokam Souop, L.~Bitjoka, Minimal isometric embeddings of graphs into
Cayley graphs of finite abelian groups, companion manuscript,
arXiv:2607.07920 [math.CO], 2026.

\bibitem{fokam-p4}
R.~Fokam Souop, L.~Bitjoka, Tight wavelet frames from isometric abelian
embeddings, companion manuscript, arXiv:2607.14168 [eess.SP], submitted to
IEEE Trans. Signal Process., 2026.

\bibitem{stanley2016}
R.~P. Stanley, Smith normal form in combinatorics, arXiv:1602.00166, 2016.

\bibitem{wilkeit1990}
E.~Wilkeit, Isometric embeddings in Hamming graphs,
J.~Combin. Theory Ser.~B 50 (1990) 179--197.

\bibitem{winkler1984}
P.~M. Winkler, Isometric embedding in products of complete graphs,
Discrete Appl. Math. 7 (1984) 221--225.

\bibitem{winkler1983}
P.~M. Winkler, Proof of the squashed cube conjecture, Combinatorica 3 (1983)
135--139.

\bibitem{fokam-comp}
R.~Fokam Souop, L.~Bitjoka, The Cayley completion of a graph, companion
manuscript, arXiv:2608.30894 [math.CO], 2026.

\bibitem{fokam-p2}
R.~Fokam Souop, L.~Bitjoka, Dimension and order bounds for isometric
embeddings of graphs into abelian Cayley graphs, and the abelian dividend,
companion manuscript, arXiv:2607.07939 [math.CO], 2026.

\bibitem{lopezrojas2016}
E.~A. Lopez-Rojas, A.~Elmir, S.~Axelsson, PaySim: a financial mobile money
simulator for fraud detection, in: Proc. 28th European Modeling and
Simulation Symposium, 2016, pp. 249--255.

\bibitem{altman2023}
E.~Altman, J.~Blanu\v{s}a, L.~von Niederh\"ausern, B.~Egressy, A.~Anghel,
K.~Atasu, Realistic synthetic financial transactions for anti-money
laundering models, in: Advances in Neural Information Processing Systems 36,
2023.

\bibitem{chen2015}
S.~Chen, R.~Varma, A.~Sandryhaila, J.~Kova\v{c}evi\'c, Discrete signal
processing on graphs: sampling theory, IEEE Trans. Signal Process. 63 (24)
(2015) 6510--6523.

\bibitem{anis2016}
A.~Anis, A.~Gadde, A.~Ortega, Efficient sampling set selection for bandlimited
graph signals using graph spectral proxies, IEEE Trans. Signal Process. 64
(14) (2016) 3775--3789.

\bibitem{ekambaram2013a}
V.~N. Ekambaram, G.~C. Fanti, B.~Ayazifar, K.~Ramchandran, Circulant
structures and graph signal processing, in: Proc. IEEE ICIP, 2013,
pp. 834--838.

\bibitem{ekambaram2013b}
V.~N. Ekambaram, G.~C. Fanti, B.~Ayazifar, K.~Ramchandran, Multiresolution
graph signal processing via circulant structures, in: Proc. IEEE DSP/SPE,
2013, pp. 112--117.

\bibitem{ekambaram2015}
V.~N. Ekambaram, G.~C. Fanti, B.~Ayazifar, K.~Ramchandran, Spline-like wavelet
filterbanks for multiresolution analysis of graph-structured data, IEEE Trans.
Signal Inf. Process. Netw. 1 (4) (2015) 268--278.

\bibitem{beck2024frames} K.~Beck, M.~Ghandehari, S.~Hudson, J.~Paltenstein,
Frames for signal processing on Cayley graphs, J.~Fourier Anal. Appl. 30 (2024)
paper no.~25, 30 pp.\ doi:10.1007/s00041-024-10128-5 (correction:
doi:10.1007/s00041-024-10137-4).
\bibitem{ghandehari2021gabor} M.~Ghandehari, D.~Guillot, K.~Hollingsworth,
Gabor-type frames for signal processing on graphs, J.~Fourier Anal. Appl. 27
(2021) paper no.~66.
\bibitem{loukas2019reduction} A.~Loukas, Graph reduction with spectral and cut
guarantees, J.~Mach. Learn. Res. 20 (2019) 1--42.
\bibitem{spielman2011sparsification} D.~A. Spielman, N.~Srivastava, Graph
sparsification by effective resistances, SIAM J.~Comput. 40 (2011) 1913--1926.
\end{thebibliography}
\end{document}